\documentclass[a4paper,UKenglish,numberwithinsect,cleveref,autoref,thm-restate]{lipics-v2021}
\hideLIPIcs  %

\usepackage[all]{xy}

\Crefname{appsec}{Appendix}{Appendices}

\usepackage{latexsym}
\usepackage{amsfonts}
\usepackage{ifthen}

\def\cC{\mathcal{C}}
\def\cD{\mathcal{D}}

\def\cP{\mathcal{P}}

\def\cR{\mathcal{R}}
\def\cS{\mathcal{S}}

\def\cV{\mathcal{V}}

\newcommand{\sort}[2][]{\ifthenelse{\equal{#1}{}}{\mbox{\sl{#2}}}{\mbox{\sl#1{#2}}}}

\newcommand{\SortBV}[1][]{\ifthenelse{\equal{#1}{}}{\sort{bv}}{\sort{bv}_{#1}}}
\newcommand{\sortBV}[1][]{\ifthenelse{\equal{#1}{}}{\sort[\scriptsize]{bv}}{\sort[\scriptsize]{bv}_{#1}}}

\newcommand{\symb}[1]{\mathsf{#1}}

\newcommand{\Pos}{\mathcal{P}os}
\newcommand{\Var}{\mathcal{V}ar}

\newcommand{\Dom}{\mathcal{D}om}

\newcommand{\mgu}[1][]{\ifthenelse{\equal{#1}{}}{\mathit{mgu}}{\mathit{mgu}_{#1}}}

\newcommand{\cStheory}[1][\mathit{theory}]{\ifthenelse{\equal{#1}{}}{\cS}{\cS_{#1}}}

\newcommand{\Sigtheory}[1][\mathit{theory}]{\ifthenelse{\equal{#1}{}}{\Sigma}{\Sigma_{#1}}}

\newcommand{\Val}[1][]{%
\ifthenelse{\equal{#1}{}}{%
{\mathcal{V}al}%
}{%
{\mathcal{V}al}_{#1}%
}}

\newcommand{\Eval}[2][]{%
\ifthenelse{\equal{#1}{}}{%
{[\![ #2 ]\!]}
}{%
{[\![ #2 ]\!]}_{#1}
}}

\newcommand{\NF}{\mathit{NF}}

\newcommand{\ito}{\xrightarrow{\mathsf{i}}{\!\!}}

\newcommand{\DP}{\mathit{DP}}

\newcommand{\SuccST}{\succ}

\newcommand{\cRfgh}{\cR_1}
\newcommand{\cPfgh}{\cP_1}
\newcommand{\cRfab}{\cR_2}
\newcommand{\cRghabc}{\cR_3}
 
\title{Termination of Innermost-Terminating Right-Linear Overlay Term Rewrite Systems} %

\titlerunning{Termination of Innermost-Terminating Right-Linear Overlay TRSs} %

\author{Naoki Nishida%
}{Graduate School of Informatics, Nagoya University, Furo-cho, Chikusa-ku, Nagoya 4648601,
Japan
\and \url{https://www.lctrs.jp/nishida/} }{nishida@i.nagoya-u.ac.jp}{https://orcid.org/0000-0001-8697-4970}{JSPS KAKENHI Grant Number JP24K02900}%

\authorrunning{N. Nishida} %

\Copyright{Naoki Nishida} %

\ccsdesc[100]{Theory of computation~Rewrite systems} %

\keywords{%
termination,
innermost termination,
dependency pair
} %

\acknowledgements{We thank the anonymous reviewers of WST~2026 not only for their valuable feedback, which improved the paper, but also suggesting several interesting future directions.}%

\nolinenumbers %

\EventEditors{}
\EventNoEds{1}
\EventLongTitle{}
\EventShortTitle{WST 2026}
\EventAcronym{WST}
\EventYear{2026}
\EventDate{}
\EventLocation{}
\EventLogo{}
\SeriesVolume{}
\ArticleNo{}

\begin{document}

\maketitle

\begin{abstract}
It has been shown that, regarding a terminating right-linear overlay term rewrite system (TRS), any rewrite sequence terminating in a normal form can be simulated by an innermost reduction.
In this paper, using this simulation property, we show that 
for a right-linear overlay TRS,
there is no infinite minimal dependency-pair chain if and only if there is no infinite innermost minimal dependency-pair chain.
As a consequence, termination and innermost termination coincide for the class of right-linear overlay TRSs.
\end{abstract}

\section{Introduction}
\label{sec:intro}

Termination is a fundamental property and remains a primary focus of research in term rewriting.
Termination is sometimes assumed for target rewrite systems, e.g., 
\emph{rewriting induction}~\cite{Red90} requires given TRSs to be terminating.
Thus, %
many powerful tools have been developed and are still improved (cf.~termCOMP~\cite{TermCOMP2015}).

Termination is the non-existence of infinite rewrite sequences, and innermost termination is the non-existence of infinite innermost rewrite sequences.
Since any innermost rewrite sequence is just a rewrite sequence,
termination implies innermost termination, but innermost termination does not, in general, imply termination;
an innermost-terminating TRS need not be terminating.
Viewed in this light, in proving innermost termination, we can use sufficient conditions for termination, together with those for innermost termination, some of which are weaker than the corresponding ones for termination.
For example, in applying the narrowing processor in the \emph{dependency pair} (DP) \emph{framework}~\cite{GTSK04,GTSKF06} to a DP problem $(\cP,\cR,e)$,%
\footnote{A DP problem is a triple $(\cP,\cR,e)$, where $\cP$ is a set of DPs, $\cR$ is a TRS $\cR$, and $e$ is a flag $e \in \{\mathbf{t},\mathbf{i}\}$ for ``\textbf{t}ermination'' and ``\textbf{i}nnermost termination''.}
if $e =\mathbf{t}$,
then right-linearity is required for a transformed DP in $\cP$, 
and otherwise (i.e., $e =\mathbf{i}$), right-linearity is not required (cf.~\cite[Definition~28]{GTSKF06}).
The reduction pair processor based on ``usable rules'' requires ``$C_E$-compatibility'' for DP problems with $\mathbf{t}$ but not for those with $\mathbf{i}$ (cf.~\cite[Theorem~17]{GTSKF06}).
On the other hand, for a class of TRSs, for which termination and innermost termination coincide, we can use proof techniques for innermost-termination in order to prove termination.
Therefore, such a class of TRSs is very interesting from both theoretical and implementation perspectives.

The largest known class of TRSs, for which termination and innermost termination coincide, is the class of locally confluent overlay TRSs, i.e.,
a locally confluent overlay TRS is terminating if and only if it is innermost terminating~\cite{Gra95}.%
\footnote{Another incompatible class is the class of right-linear overlay TRSs, which has been shown in an unpublished note~\cite{Sakai03note}.}
Being an overlay system is a decidable syntactic property%
\footnote{"Syntactic properties of a TRS" refer to properties determined solely by structural characteristics of the TRS and not depending on the reduction of the TRS.}
of TRSs, but local confluence is not, while there are some syntactic sufficient conditions for local confluence of TRSs, e.g., non-overlappingness.
Note that to show local confluence, we sometimes first prove termination and then joinability of critical pairs.
Identifying classes of TRSs where innermost termination implies termination is therefore of significant interest.

In this paper, we show that for a right-linear overlay TRS, there is no infinite minimal dependency-pair chain if and only if there is no infinite innermost minimal dependency-pair chain, 
where dependency-pair chains are sequences of dependency pairs of the right-linear overlay TRS.
To be more precise, we show that
for a set $\cP$ of DPs such that $\cR\cup\cP$ is a right-linear overlay TRS, there is no infinite minimal $(\cP,\cR)$-chain if and only if there is no infinite innermost minimal $(\cP,\cR)$-chain.
As a consequence, 
termination and innermost termination coincide for the class of right-linear overlay TRSs (cf.~\cite{Sakai03note}).

Our main result is based on the proof of the following claim:
Regarding a terminating right-linear overlay TRS, any rewrite sequence terminating in a normal form can be simulated by an innermost reduction~\cite{SOS03}.
To prove our main goal ``innermost termination implies termination'', given an infinite rewrite sequence of an innermost-terminating right-linear overlay TRS $\cR$, we would like to show the existence of an infinite innermost rewrite sequence of $\cR$.
Unfortunately, for this approach, we cannot use the result in~\cite{SOS03}, which only performs for finite rewrite sequences terminating in normal forms.
On the other hand, for a (possibly infinite) minimal $(\cP,\cR)$-chain, all terms in the chain are terminating w.r.t.\ $\cR$.
Thus, using an auxiliary lemma for the main result in~\cite{SOS03}, we show that for a finite minimal $(\cP,\cR)$-chain of length $n$, there exists an innermost minimal $(\cP,\cR)$-chain of length $n$, provided that $\cR\cup\cP$ is a right-linear overlay system.
Note that the length of a $(\cP,\cR)$-chain is defined by the number of DPs included.
Using this auxiliary property, our proof for the main result proceeds by contradiction as follows:
\begin{enumerate}
    \item Assume, for the sake of contradiction, that there is no infinite innermost minimal $(\cP,\cR)$-chain and there is an infinite minimal $(\cP,\cR)$-chain starting with a term $s^\#$.
    \item Let $n$ be the maximum length of innermost minimal $(\cP,\cR)$-chain starting with $s^\#$.
    \item Construct an innermost minimal $(\cP,\cR)$-chain of length $n+1$ from a finite minimal $(\cP,\cR)$-chain of the same length, where these chains start with $s^\#$.
    \item This contradicts the assumption about $n$.
\end{enumerate}
Note that the above approach relies on termination of proper subterms in the minimal $(\cP,\cR)$-chain, but not the property of being, e.g., finitely branching.

Regarding our main result %
on dependency-pair chains, we do not assume that a given TRS $\cR$ is innermost terminating.
Thus, the result provides a DP processor %
that converts a given DP problem $(\cP,\cR,\mathbf{t})$ to $(\cP,\cR,\mathbf{i})$ if $\cR\cup\cP$ is a right-linear overlay system.
In the DP framework~\cite{GTSK04,GTSKF06}, some DP processors, such as the narrowing processor, transform DPs into rewrite rules that are not DPs of given TRSs.
To use the switching processor mentioned above in the DP framework, we do not restrict $\cP$ in our main result %
to be sets of DPs of given TRSs, and $\cP$ is assumed to be a non-collapsing right-linear TRS such that $\cR\cup\cP$ is an overlay system.

\section{Preliminaries}
\label{sec:preliminaries}

In this section, we briefly recall some syntactic properties of TRSs and some essential results on innermost rewriting and dependency pairs.
Familiarity with basic notions and notations on term rewriting, rewrite strategies, and dependency pairs~\cite{BN98,TRS,Ohl02} is assumed.

The reduction of a \emph{term rewrite system} (TRS, for short) $\cR$ is denoted by $\to_\cR$:
$s \to_\cR t$ if and only if there exist a rewrite rule $\ell \to r \in \cR$, a position $p$ of $s$, and a substitution $\theta$
such that $s|_p = \ell\theta$ and $t = s[r\theta]_p$.
We often write $s \to_{p,\cR} t$ instead of $s \to_\cR t$, and 
write $s \to_{>\varepsilon,\cR} t$ if $p > \varepsilon$.
The \emph{innermost reduction} of $\cR$ is denoted by $\ito_\cR$:
$s \mathrel{\ito_\cR} t$ (or $s \mathrel{\ito_{p,\cR}} t$) if and only if $s \to_{p,\cR} t$
and $s|_p$ is an innermost redex of $\cR$ (i.e., every proper subterm of $s|_p$ is a normal form of $\cR$).
A sequence $s_0 \to_\cR s_1 \to_\cR \cdots$ is called a \emph{rewrite sequence} of $\cR$.
A sequence $s_0 \mathrel{\ito_\cR} s_1 \mathrel{\ito_\cR} \cdots$ is called an \emph{innermost rewrite sequence} of $\cR$.
A term $t$ is said to be \emph{terminating} (resp.\ \emph{innermost terminating}) w.r.t.\ a TRS
if there is no infinite (resp.\ innermost) rewrite sequence of the TRS, which starts from $t$.
The set of normal forms of $\cR$ over $\Sigma$ is denoted by $\NF_\cR(\Sigma,\cV)$.
A rewrite rule $\ell \to r$ is called \emph{right-linear} if the right-hand side $r$ is linear.
Rewrite rule $\ell \to r$ is called \emph{collapsing} if the right-hand side $r$ is a variable.
Note that the right-hand side of a non-collapsing rule is not a variable.
A TRS $\cR$ is called \emph{right-linear} (resp.\ \emph{non-collapsing}) if all rewrite rules in $\cR$ are right-linear (resp.\ non-collapsing).
TRS $\cR$ is said to \emph{have an inner-overlap} if there exist rewrite rules $\ell \to r$ and $\ell' \to r'$ in $\cR$ such that 
a renamed proper non-variable subterm $s$ of $\ell$ (i.e., $\Var(s)\cap\Var(\ell')=\emptyset$) is unifiable with $\ell'$.
TRS $\cR$ is called an \emph{overlay system} if $\cR$ has no inner-overlap.

Regarding a right-linear overlay TRS $\cR$, every rewrite sequence terminating in a normal form of $\cR$ can be simulated by an innermost reduction of $\cR$.
\begin{theorem}[\cite{SOS03}]
\label{thm:SOS03-main}
Let $\cR$ be a terminating right-linear overlay TRS over a signature $\Sigma$.
For all terms $s,t \in T(\Sigma,\cV)$,
if $s \to_\cR^* t \in \NF_\cR(\Sigma,\cV)$, 
then $s \mathrel{\ito_\cR^*} t$.
\end{theorem}
Innermost termination of locally confluent overlay TRSs implies termination.
\begin{theorem}[\cite{Gra95}]
\label{thm:SIN-SN-for-LCR-overlay-systems}
A locally confluent overlay TRS $\cR$ is terminating if and only if it is innermost terminating.
\end{theorem}
The class of right-linear overlay TRSs is incomparable with that of locally confluent overlay TRSs.
\begin{example}
\label{ex:RLOS-LCROS-incomparable}
The TRS 
$\{
\symb{0} + y \to y, ~ 
\symb{s}(x) + y \to \symb{s}(x + y), ~
\symb{0} \times y \to \symb{0}, ~
\symb{s}(x) \times y \to (x \times y) + y
\}$ is orthogonal (and thus a locally-confluent overlay system) yet it is not right-linear.
The TRS 
$\{
\symb{a} \to \symb{b}, ~
\symb{a} \to \symb{c}
\}$ is a right-linear overlay system but not locally confluent.
\end{example}

Let $\cR$ be a TRS over a signature $\Sigma$.
We denote the set of \emph{defined symbols} and \emph{constructors} of $\cR$ by $\cD_\cR$ and $\cC_\cR$, respectively.
The tuple symbol of a defined symbol $f \in \cD_\cR$ is denoted by $f^\#$.
The set of tuple symbols for $\cD_\cR$ is denoted by $\cD^\#_\cR$.
If $t=f(t_1,\ldots,t_n)$ with $f\in\cD_\cR$, then $f^\#(t_1,\ldots,t_n)$ is denoted by $t^\#$.
For each rule $\ell\to r \in\cR$, a rewrite rule $\ell^\#\to t^\#$ is called a \emph{dependency pair} (DP, for short) of $\cR$ if $t$ is a subterm of $r$ and $root(t)\in\cD_\cR$~\cite{AG00}.
The set of DPs of $\cR$ is denoted by $\DP(\cR)$.
Note that $\DP(\cR)$ is a non-collapsing overlay TRS over the signature $\Sigma\cup\cD_\cR^\#$.
Let $\cP \subseteq \DP(\cR)$.
A sequence $s_1^\# \to t_1^\#, s_2^\# \to t_2^\#,\ldots$ %
of DPs in $\cP$ is called a \emph{dependency-pair chain} of $\cP$ w.r.t.\ $\cR$ ($(\cP,\cR)$-chain, for short) 
if there are substitutions $\sigma_1,\sigma_2,\ldots$ such that $t_i^\#\sigma_i\to_\cR^*s_{i+1}^\#\sigma_{i+1}$ for each $i > 0$.
The rewrite sequence $s_1^\#\sigma_1 \to_{\cP} t_1^\#\sigma_1 \to_\cR^* s_2^\#\sigma_2 \to_{\cP} t_2^\#\sigma_2 \to_\cR^* \cdots$ is also called a $(\cP,\cR)$-chain.
The $(\cP,\cR)$-chain above %
is called \emph{minimal} if all $t_1^\#\sigma_1,t_2^\#\sigma_2,\ldots$ are terminating w.r.t.\ $\cR$.
The $(\cP,\cR)$-chain above %
is called \emph{innermost} if $t_i^\#\sigma_i\ito_\cR^*s_{i+1}^\#\sigma_{i+1}$ and $s_{i+1}^\#\sigma_{i+1} \in \NF_\cR(\Sigma\cup\cD^\#_\cR,\cV)$ for all $i > 0$.

\begin{theorem}[\cite{AG00}]
\label{thm:DP-criteria}
Let $\cR$ be a TRS.
Then, both of the following statements hold:
\begin{itemize}
    \item $\cR$ is terminating if and only if there is no minimal $(\DP(\cR),\cR)$-chain, and
    \item $\cR$ is innermost terminating if and only if there is no minimal innermost $(\DP(\cR),\cR)$-chain.
\end{itemize}
\end{theorem}

\section{From Innermost Termination to Termination}
\label{sec:from-innermost-termination-to-termination}

In this section, for a right-linear overlay TRS $\cR$ over a signature $\Sigma$, we show the main claim:
    For a set $\cP$ of DPs such that $\cR\cup\cP$ is a right-linear overlay TRS, there is no infinite minimal $(\cP,\cR)$-chain if and only if there is no infinite innermost minimal $(\cP,\cR)$-chain (\Cref{thm:inf-DP-chain-iff-inf-innermost-DP-chain}).
The \textit{only-if} part of the main claim trivially holds because any innermost minimal $(\cP,\cR)$-chain is a minimal $(\cP,\cR)$-chain.
Thus, in the rest of this section, we focus on proving the \textit{if} part of the main claim.
Thanks to \Cref{thm:DP-criteria}, as a corollary, the main claim establishes that 
termination and innermost termination coincide for the class of right-linear overlay TRSs.

As described in \Cref{sec:intro}, we generalize our main claim so that we do not restrict chains to sequences of DPs (\Cref{thm:inf-chain-iff-inf-innermost-chain}).
To be more precise, we let $\cP$ in the main claim be a non-collapsing TRS over an extended signature of $\Sigma$.
For this generalization, we adapt the notion of chains of DPs to non-collapsing TRSs.
Hereafter, we will abuse the terminology ``chain'' because chains of DPs defined in \Cref{sec:preliminaries} are special cases of chains defined below.

Let $\cR$ be a TRS over a signature $\Sigma$, and $\cP$ be a non-collapsing TRS over an extended signature $\Sigma'$ of $\Sigma$.
We call a (possibly infinite) sequence $t_0 \to_{>\varepsilon,\cR}^* s_1 \to_{\varepsilon,\cP} t_1 \to_{>\varepsilon,\cR}^* s_2 \to_{\varepsilon,\cP} t_2 \to_{>\varepsilon,\cR}^* \cdots$ a \emph{$(\cP,\cR)$-chain}.
Note that $t_0,s_1,t_1,\ldots$ are over $\Sigma'$ and $\cR$ is considered a TRS over $\Sigma'$.
The $(\cP,\cR)$-chain is called \emph{minimal} if all proper subterms of $t_0,t_1,t_2,\ldots$ are terminating w.r.t.\ $\cR$.
Note that $t_i$ of the minimal $(\cP,\cR)$-chain may be non-terminating w.r.t.\ $\cR$.
The $(\cP,\cR)$-chain is called \emph{innermost} if 
$t_i \mathrel{\ito_{>\varepsilon,\cR}^*} s_{i+1}$
and $s_{i+1} \in \NF_\cR(\Sigma',\cV)$ for all $i \geq 0$.
Note that a finite $(\cP,\cR)$-chain is written in the form 
$t_0 \to_{>\varepsilon,\cR}^* s_1 \to_{\varepsilon,\cP} t_1 \to_{>\varepsilon,\cR}^* s_2 \to_{\varepsilon,\cP} t_2 \to_{>\varepsilon,\cR}^* \cdots \to_{>\varepsilon,\cR}^* s_n \to_{\varepsilon,\cP} t_n \to_{>\varepsilon,\cR}^* t'_n$.
Note also that a finite innermost $(\cP,\cR)$-chain is written in the form 
$t_0 \mathrel{\ito_{>\varepsilon,\cR}^!} s_1 \to_{\varepsilon,\cP} t_1 \mathrel{\ito_{>\varepsilon,\cR}^!} s_2 \to_{\varepsilon,\cP} t_2 \mathrel{\ito_{>\varepsilon,\cR}^!} \cdots \mathrel{\ito_{>\varepsilon,\cR}^!} s_n \to_{\varepsilon,\cP} t_n \mathrel{\ito_{>\varepsilon,\cR}^*} t'_n$.
where 
${\to_\cR^!} = {\to_\cR^*} \cap \{ (s,t) \mid t \in \NF_\cR(\Sigma',\cV) \}$
and
${\ito_\cR^!} = {\ito_\cR^*} \cap \{ (s,t) \mid t \in \NF_\cR(\Sigma',\cV) \}$.
The length of a finite $(\cP,\cR)$-chain is defined as the number of $\to_{\varepsilon,\cP}$-steps included.
We write $s \mathrel{\to_{>\varepsilon,\cR}^!} t$ (resp.\ $s \mathrel{\ito_{>\varepsilon,\cR}^!} t$) if 
either $s$ is a variable or 
$s = f(s_1,\ldots,s_n) \mathrel{\to_\cR^*} f(t_1,\ldots,t_n) = t$
and 
$s_i \mathrel{\to_\cR^!} t_i$ (resp.\ $s_i \mathrel{\ito_\cR^!} t_i$) for any $i \in \{1,\ldots,n\}$.

In proving the \textit{if} part of the generalized main claim, innermost termination of $\cR$ is not assumed, while all proper subterms of $t_0,t_1,\ldots$ of the $(\cP,\cR)$-chain $t_0 \to_{>\varepsilon,\cR}^* s_1 \to_{\varepsilon,\cP} t_1 \to_{>\varepsilon,\cR}^* s_2 \to_{\varepsilon,\cP} t_2 \to_{>\varepsilon,\cR}^* \cdots$ are assumed to be terminating w.r.t.\ $\cR$.
For this reason, referring to an auxiliary lemma in~\cite[Lemma~3.4]{SOS03}, we prepare the following lemma without assuming either termination or innermost termination of $\cR$.

\begin{restatable}{lemma}{LemOlRlIsnImpliesSn}
\label{lem:OL-RL-ISN-implies-SN}
Let $\cR$ be a right-linear overlay TRS over a signature $\Sigma$, 
$s$ be a linear term in $T(\Sigma,\cV)$,
$t$ be a normal form of $\cR$ over $\Sigma$ (i.e., $t \in \NF_\cR(\Sigma,\cV)$),
and
$\sigma$ be a substitution
such that $s\sigma$ is terminating w.r.t.\ $\cR$. %
If $s\sigma \mathrel{\to_\cR^!} t$, 
then there exists a substitution $\sigma'$ such that
\begin{itemize}
    \item
    $\Dom(\sigma') = \Dom(\sigma|_{\Var(s)})$,
    \item 
    $x\sigma \mathrel{\to_\cR^*} x\sigma' \in \NF_\cR(\Sigma,\cV)$ for all variables $x \in \Var(s)$
    (i.e., $s\sigma \mathrel{\to_\cR^*} s\sigma'$),
        and
    \item 
    $s\sigma' \mathrel{\ito_\cR^!} t$.
\end{itemize}
\end{restatable}

\noindent
The proof of \Cref{lem:OL-RL-ISN-implies-SN} is provided in \Cref{sec:LemOlRlIsnImpliesSn}.
The main difference between \Cref{lem:OL-RL-ISN-implies-SN} and~\cite[Lemma~3.4]{SOS03} is that 
the latter assumes termination of $\cR$, the former does not, and assumes termination of $s\sigma$.

A variant of \Cref{thm:SOS03-main} such that termination of a term $s$ is assumed instead of termination of a TRS $\cR$ is obtained from \Cref{lem:OL-RL-ISN-implies-SN}.

\begin{restatable}{lemma}{LemReductionSimulatedbyInnermostOne}%
\label{lem:reduction-simulated-by-innermost-one}
Let $\cR$ be a right-linear overlay TRS over a signature $\Sigma$, $s$ be a terminating term in $T(\Sigma,\cV)$, and $t$ be a normal form of $\cR$ over $\Sigma$ (i.e., $t \in \NF_\cR(\Sigma,\cV)$).
If $s \mathrel{\to_\cR^!} t$, %
then $s \mathrel{\ito_\cR^!} t$.
\end{restatable}
\begin{proof}
Let $\Pos_\cV(s) = \{p_1,\ldots,p_m\}$.
Then, $s$ is of the form $s[x_1,\ldots,x_m]_{p_1,\ldots,p_m}$.
Note that $\Var(s) = \{x_1,\ldots,x_m\}$.%
\footnote{There may be positions $p_i,p_j$ such that $i\ne j$, $s|_{p_i}=x_i=x_j=s|_{p_j}$.}
Let $y_1,\ldots,y_m$ be pairwise distinct variables.
Then, the term $s[y_1,\ldots,y_m]_{p_1,\ldots,p_m}$ is linear.
Let $s'=s[y_1,\ldots,y_m]_{p_1,\ldots,p_m}$ and $\sigma=\{ y_1 \mapsto x_1, \ldots, y_m \mapsto x_m\}$.
Then, we have that 
    $s'\sigma = s$,
        and
    $y_i\sigma = x_i \in \cV \subseteq \NF_\cR(\Sigma,\cV)$ for all $1 \leq i \leq m$.
Thus, we have that $s'\sigma \mathrel{\to_\cR^!} t$.
It follows from \Cref{{lem:OL-RL-ISN-implies-SN}} that there exists a substitution $\sigma'$ such that
\begin{itemize}
    \item $\Dom(\sigma') = \Dom(\sigma)$,
    \item $x\sigma \mathrel{\to_\cR^*} x\sigma' \in \NF_\cR(\Sigma,\cV)$ for all variables $x \in \Var(s')$,
        and
    \item $s'\sigma' \mathrel{\ito_\cR^!} t$.
\end{itemize}
For any variable $x \in \Dom(\sigma)$,
since $x\sigma \in \NF_\cR(\Sigma,\cV)$, we have that $x\sigma=x\sigma'$.
Therefore, we have that 
$s = s'\sigma = s'\sigma' \mathrel{\ito_\cR^!} t$.
\end{proof}

Next, we show an auxiliary lemma, which is a key property to prove our generalized main claim (\Cref{thm:inf-chain-iff-inf-innermost-chain} below).

\begin{restatable}{lemma}{LemExistsTheSameLengthInnermostChain}
\label{lem:exists-the-same-length-innermost-chain}
Let $\cR$ be a TRS over a signature $\Sigma$,
$\cP$ be a non-collapsing TRS over an extended signature $\Sigma'$ of $\Sigma$,
$t_0$ be a term in $T(\Sigma',\cV)$,
$u$ be a non-variable term in $T(\Sigma',\cV)$,
$\sigma,\sigma'$ be substitutions,
and
$t_0 \mathrel{(\to_{>\varepsilon,\cR}^*\cdot\to_{\varepsilon,\cP})^n} u\sigma 
\mathrel{\to_{>\varepsilon,\cR}^*}
u\sigma'$ be a finite minimal $(\cP,\cR)$-chain such that
\begin{itemize}
    \item $x\sigma \mathrel{\to_\cR^*} x\sigma' \in \NF_\cR(\Sigma',\cV)$ for all variables $x \in \Var(u)$,
        and
    \item all proper subterms of $u\sigma'$ are normal forms of $\cR$ over $\Sigma'$.
\end{itemize}
Suppose that $\cR\cup\cP$ is a right-linear overlay system.
Then, there exists an innermost minimal $(\cP,\cR)$-chain 
$t_0 \mathrel{(\ito_{>\varepsilon,\cR}^!\cdot\ito_{\varepsilon,\cP})^n\cdot\ito_{>\varepsilon,\cR}^*} u\sigma'$.
\end{restatable}
\begin{proof}
We prove this claim by induction on $n$.

We first consider the case where $n=0$.
We have that $t_0 = u\sigma$, and thus $t_0 = u\sigma \mathrel{\to_{>\varepsilon,\cR}^*} u\sigma'$.
It follows from \Cref{lem:reduction-simulated-by-innermost-one} that $x\sigma \mathrel{\ito_\cR^*} x\sigma'$ for all variables $x \in \Var(u)$,
and thus we have that $t_0=u\sigma \mathrel{\ito_{>\varepsilon,\cR}^*} u\sigma'$, which is an innermost minimal $(\cP,\cR)$-chain of length~$0$.

Next, we consider the remaining case where $n>0$.
The overview of the proof for this case can be seen in \Cref{fig:overview-of-main-case-of-Main-Lemma}.
Suppose that
\[
t_0 \mathrel{(\to_{>\varepsilon,\cR}^*\cdot\to_{\varepsilon,\cP})^{n-1}} 
\cdot
\mathrel{\to_{>\varepsilon,\cR}^*}
\ell\delta
\mathrel{\to_{\varepsilon,\cP}}
r\delta
=
u\sigma
\mathrel{\to_{>\varepsilon,\cR}^*}
u\sigma'
\]
where $\ell \to r \in \cP$.
Since $\cP$ is non-collapsing, $r$ is not a variable.
By assumption, $u$ is not a variable.
Let $r = f(r_1,\ldots,r_m)$ and $u = f(u_1,\ldots,u_m)$.
Since the chain is minimal, all $r_1\delta,\ldots,r_m\delta$ are terminating w.r.t.\ $\cR$.
By assumption, all $u_1\sigma',\ldots,u_m\sigma'$ are normal forms of $\cR$.
It follows from \Cref{lem:OL-RL-ISN-implies-SN} that
for each $i \in \{1,\ldots,m\}$, there exists a 
substitution $\delta_i$ such that 
\begin{itemize}
    \item $\Dom(\delta_i) = \Dom(\delta|_{\Var(r_i)})$, 
    \item $x\delta \to_\cR^* x\delta_i \in \NF_\cR(\Sigma',\cV)$ for all variables $x \in \Var(r_i)$,
    and
    \item $r_i\delta_i \mathrel{\ito_\cR^!} u_i\sigma'$,
\end{itemize}
and thus $r_i\delta \to_\cR^* r_i\delta_i$ for all $i \in \{1,\ldots,m\}$.
Since $\cP$ is right-linear, $r$ is linear, and thus all $r_1,\ldots,r_m$ are linear and $\Var(r_i) \cap \Var(r_j) = \emptyset$ for each $i,j \in \{1,\ldots,m\}$ with $i \ne j$.
Let $\delta'=\bigcup_{i=1}^m \delta_i$.
Then, $\delta'$ is a substitution and we have that
\[
r\delta = f(r_1\delta,\ldots,r_m\delta) \mathrel{\to_{>\varepsilon,\cR}^*} 
f(r_1\delta',\ldots,r_m\delta') 
=
r\delta'
\mathrel{\ito_{>\varepsilon,\cR}^*}
f(u_1\sigma',\ldots,u_m\sigma') 
= u\sigma'
\]

We now consider the variables in $\ell$, which does not appear in $r$.
Let $x$ be one of such variables.
Since the chain is minimal, $x\delta$ has a normal form.
Let $u_x$ denote an arbitrary normal form of $x\delta$, i.e., $x\delta \mathrel{\to_\cR^!} u_x \in \NF_\cR(\Sigma',\cV)$.
Let $\delta'' = \delta'\cup\{ x \mapsto u_x \mid x \in \Var(\ell)\setminus\Var(r), x\delta \mathrel{\to_\cR^!} u_x \}$.
Then, we have that $\ell\delta \mathrel{\to_\cR^*} \ell\delta''$.
Since $\cP$ is a TRS, $\ell$ is not a variable.
Hence, we have that $\ell\delta \mathrel{\to_{>\varepsilon,\cR}^*} \ell\delta'' \to_{\varepsilon,\cP} r\delta'' = r\delta'$, and thus
\[
t_0 \mathrel{(\to_{>\varepsilon,\cR}^*\cdot\to_{\varepsilon,\cP})^{n-1}} 
\cdot
\mathrel{\to_{>\varepsilon,\cR}^*}
\ell\delta
\mathrel{\to_{>\varepsilon,\cR}^*} \ell\delta'' 
\]
such that $x\delta \mathrel{\to_\cR^!} x\delta''$ for all variables $x \in \Var(\ell)$.
Since $\cR\cup\cP$ is an overlay system, all proper subterms of $\ell\delta''$ are normal forms of $\cR$.%
\footnote{
Suppose that there exists a proper subterm of $\ell\delta''$, which is not a normal form of $\cR$.
Then, $\ell$ has a proper subterm $\ell'$ such that $\ell'\delta''$ is a redex of $\cR$.
Since $x\delta'' \in \NF_\cR(\Sigma',\cV)$ for all variables $x \in \Var(\ell)$, $\ell'$ is not a variable.
There exists a rule $\ell'' \to r'' \in \cR$ such that $\ell'\delta'' = \ell''\theta$ for some substitution $\theta$.
This contradicts the assumption that $\cR\cup\cP$ is an overlay system.
}
Since the above sequence is a minimal $(\cP,\cR)$-chain of length $n-1$,
by the induction hypothesis, 
we have that 
$t_0 \mathrel{(\ito_{>\varepsilon,\cR}^!\cdot\ito_{\varepsilon,\cP})^{n-1}} 
\cdot
\mathrel{\ito_{>\varepsilon,\cR}^!} 
\ell\delta'' 
$, and thus we have the following innermost minimal $(\cP,\cR)$-chain of length $n$:
\[
t_0 \mathrel{(\ito_{>\varepsilon,\cR}^!\cdot\ito_{\varepsilon,\cP})^{n-1}} 
\cdot
\mathrel{\ito_{>\varepsilon,\cR}^!} 
\ell\delta'' 
\mathrel{\ito_{\varepsilon,\cP}}
r\delta'' = r\delta'
\mathrel{\ito_{>\varepsilon,\cR}^*}
u\sigma'
\]
\end{proof}

\begin{figure}[t]
\hfil
\xymatrix@R=28pt@C=6pt{
t_0 \ar@{}[rrrrrrrrrrr]|{\mbox{$(\to_{>\varepsilon,\cR}^*\cdot\to_{\varepsilon,\cP})^{n-1} \cdot {\to_{>\varepsilon,\cR}^*}$}}
\ar@{-->}^{\mathsf{i}}@/_12pt/[drrrrrrrrrrr]
\ar@{}[drrrrrrrrrrr]|(.6){\mbox{\small I.H.}}
&&&&&&&&&&& \ell\delta \ar[rr]_(.6){\varepsilon,\cP} \ar@{..>}[d]_{*}^(.7){>\varepsilon,\cR}
&& r\delta \ar@{}[r]|{\mbox{$=$}} \ar@{..>}[d]_\ast^(.7){>\varepsilon,\cR}
& u\sigma \ar[rrrr]^\ast_(.65){>\varepsilon,\cR} 
&&&& u\sigma'
\\
&&&&&&&&&&& \ell\delta'' \ar[rr]^{\mathsf{i}}_(.6){\varepsilon,\cP}
&& r\delta'' \ar@{}[r]|{\mbox{$=$}}
& r\delta'
\ar@{..>}@/_15pt/[urrrr]^(.4){\mathsf{i}}|(.5){\raisebox{8pt}{\scriptsize !}}_(.8){>\varepsilon,\cR} \ar@{}[urrrr]|{\mbox{\raisebox{5pt}{\small \Cref{lem:OL-RL-ISN-implies-SN}~~}}}
\\
}
\caption{An overview of the proof for the case where $n>0$ in \Cref{lem:exists-the-same-length-innermost-chain}, where solid arrows represent assumptions, dotted arrows represent consequences, and the dashed arrow with $\mathsf{i}$ represents
$(\ito_{>\varepsilon,\cR}^!\cdot\ito_{\varepsilon,\cP})^{n-1} \cdot {\ito_{>\varepsilon,\cR}^!}$ for consequences.}
\label{fig:overview-of-main-case-of-Main-Lemma}
\end{figure}

The following example shows that without assuming $\cP$ non-collapsing, \Cref{lem:exists-the-same-length-innermost-chain} does not hold.
\begin{example}
\label{ex:counter-ex-of-wo-non-collapsing}
Let us consider the following two TRSs $\cPfgh, \cRfgh$:
\[
\cPfgh = 
\{~
\symb{f}(x) \to x, ~~
\symb{g}(x) \to \symb{f}(\symb{g}(x))
~\}
\qquad
\cRfgh =
\{~
\symb{g}(x) \to \symb{h}(x)
~\}
\]
The first rule of $\cPfgh$ is collapsing and $\cPfgh\cup\cRfgh$ is a right-linear overlay TRS.
We have the minimal $(\cPfgh,\cRfgh)$-chain
$\symb{f}(\symb{g}(x)) \to_{\varepsilon,\cPfgh} \symb{g}(x) \to_{\varepsilon,\cPfgh} \symb{f}(\symb{g}(x))$ of length $2$.
On the other hand, 
there is no innermost minimal $(\cPfgh,\cRfgh)$-chain of length $2$, which starts with $\symb{f}(\symb{g}(x))$;
the longest innermost minimal $(\cPfgh,\cRfgh)$-chain starting with $\symb{f}(\symb{g}(x))$ is
$\symb{f}(\symb{g}(x)) \mathrel{\ito_{>\varepsilon,\cRfgh}} \symb{f}(\symb{h}(x)) \mathrel{\ito_{\varepsilon,\cPfgh}} \symb{h}(x)$, which is of length $1$.
\end{example}

\Cref{lem:exists-the-same-length-innermost-chain} implies the generalized main claim.

\begin{theorem}
\label{thm:inf-chain-iff-inf-innermost-chain}
Let $\cR$ be a TRS over a signature $\Sigma$, and
$\cP$ be a non-collapsing TRS over an extended signature of $\Sigma$
such that $\cR\cup\cP$ is a right-linear overlay system.
Then,
there exists no infinite minimal $(\cP,\cR)$-chain if and only if there exists no infinite innermost minimal $(\cP,\cR)$-chain.
\end{theorem}
\begin{proof}
The \textit{only-if} part is trivial because 
any infinite innermost minimal $(\cP,\cR)$-chain is an infinite minimal $(\cP,\cR)$-chain.
We prove the \textit{if} part by contradiction.
Suppose that there exists no infinite innermost minimal $(\cP,\cR)$-chain and there exists an infinite minimal $(\cP,\cR)$-chain.
Let $s_1 \mathrel{\to_{\varepsilon,\cP}} t_1 \mathrel{\to_{>\varepsilon,\cR}^*} s_2 \mathrel{\to_{\varepsilon,\cP}} t_2 \mathrel{\to_{>\varepsilon,\cR}^*} \cdots$ be an infinite minimal $(\cP,\cR)$-chain.
Since there is no infinite minimal $(\cP,\cR)$-chain starting from $s_1$ by the assumption, there exists the maximum length of innermost minimal $(\cP,\cR)$-chain starting from $s_1$.
Let $n$ be the maximum length of innermost minimal $(\cP,\cR)$-chains starting from $s_1$.
Let us consider the finite prefix chain of length $n+1$:
$s_1 \mathrel{\to_{\varepsilon,\cP}} t_1 \mathrel{\to_{>\varepsilon,\cR}^*} s_2 \mathrel{\to_{\varepsilon,\cP}} t_2 \mathrel{\to_{>\varepsilon,\cR}^*} \cdots \mathrel{\to_{>\varepsilon,\cR}^*} s_{n+1} \mathrel{\to_{\varepsilon,\cP}} t_{n+1}$.
Then, it follows from \Cref{lem:exists-the-same-length-innermost-chain} that 
there exists an innermost minimal $(\cP,\cR)$-chain of length $n+1$.
This contradicts the assumption that $n$ is the maximum length of the finite innermost minimal $(\cP,\cR)$-chains starting from $s_1$.
Therefore, the \textit{if} part holds.
\end{proof}
As for \Cref{lem:exists-the-same-length-innermost-chain},
without assuming $\cP$ non-collapsing, \Cref{thm:inf-DP-chain-iff-inf-innermost-DP-chain} does not hold
(see \Cref{ex:counter-ex-of-wo-non-collapsing}).

By definition, it is clear that 
if $\cR$ is a right-linear overlay TRS, then $\cR\cup\DP(\cR)$ is so.
\begin{proposition}
\label{prop:RL-OVL-implies-RL-OVL-for-DPs}
Let $\cR$ be a right-linear overlay TRS over a signature $\Sigma$.
Then, for any subset $\cP$ of $\DP(\cR)$ (i.e., $\cP \subseteq \DP(\cR)$),
$\cR\cup\cP$ is a right-linear overlay TRS.
\end{proposition}
\begin{proof}
Trivial by definition.
\end{proof}

Since $\DP(\cR)$ is a non-collapsing TRS, by \Cref{prop:RL-OVL-implies-RL-OVL-for-DPs},  we can choose any subset of $\DP(\cR)$ as $\cP$ in \Cref{thm:inf-chain-iff-inf-innermost-chain}.

\begin{theorem}
\label{thm:inf-DP-chain-iff-inf-innermost-DP-chain}
Let $\cR$ be a TRS, %
and $\cP \subseteq \DP(\cR)$ such that 
$\cR\cup\cP$ is a right-linear overlay system.
Then,
there exists no infinite minimal $(\cP,\cR)$-chain if and only if there exists no infinite innermost minimal $(\cP,\cR)$-chain.
\end{theorem}
\begin{proof}
An immediate consequence of \Cref{thm:inf-chain-iff-inf-innermost-chain,prop:RL-OVL-implies-RL-OVL-for-DPs}.
\end{proof}

The following claim is an immediate consequence of \Cref{thm:DP-criteria,thm:inf-DP-chain-iff-inf-innermost-DP-chain}.

\begin{restatable}{theorem}{ThmSNiffSINforRLOverlayTRSs}
\label{thm:SN-iff-SIN-for-RL-overlay-TRSs}
A right-linear overlay TRS is terminating
if and only if
it is innermost terminating.
\end{restatable}
\begin{proof}
Since $\cR$ is a right-linear overlay system,
by \Cref{prop:RL-OVL-implies-RL-OVL-for-DPs},
$\cR\cup\DP(\cR)$ is a right-linear overlay system.
It follows from \Cref{thm:inf-chain-iff-inf-innermost-chain} that
there exists no infinite minimal $(\DP(\cR),\cR)$-chain if and only if there exists no infinite innermost minimal $(\DP(\cR),\cR)$-chain.
Therefore, by \Cref{thm:DP-criteria}, the claim holds.
\end{proof}
One may think that \Cref{lem:reduction-simulated-by-innermost-one} entails the \textit{if} part of the claim in \Cref{thm:SN-iff-SIN-for-RL-overlay-TRSs}.
In fact, \Cref{lem:reduction-simulated-by-innermost-one} made us conjecture the \textit{if} part.
To prove the \textit{if} part, it would be usual to use contradiction by constructing an infinite innermost rewrite sequence from an infinite rewrite sequence.
However, \Cref{lem:reduction-simulated-by-innermost-one} cannot be applied to infinite rewrite sequences in order to construct infinite innermost rewrite sequences.
On the other hand, all proper subterms in minimal $(\cP,\cR)$-chains are terminating w.r.t.\ $\cR$ and thus, \Cref{lem:reduction-simulated-by-innermost-one} can be applied locally (i.e., $\to_{>\varepsilon,\cR}^*$ steps between $\to_{\varepsilon,\cP}$ steps).

Finally, we show some examples to explain our assumption, i.e., the necessity of ``right-linearity and an overlay system'', while it is not necessary condition.
\begin{example}
Let us consider the following TRS, which is right-linear, confluent, but not an overlay system:
\[
\cRfab =
\{~
\symb{f}(\symb{a}) \to \symb{f}(\symb{a}), ~~
\symb{a} \to \symb{b} 
~\}
\]
This TRS is innermost terminating but not terminating because we have an infinite non-innermost rewrite sequence
$\symb{f}(\symb{a}) \to_{\cRfab} \symb{f}(\symb{a}) \to_{\cRfab} \cdots$.
There is no innermost rewrite sequence that simulates the infinite rewrite sequence because of the inner-overlap between the first and second rules of $\cRfab$.
\end{example}

\begin{example}
Let us consider the following TRS, which is an overlay system, but neither right-linear nor locally confluent:
\[
\cRghabc =
\{~
\symb{g}(x) \to \symb{h}(x,x), ~~
\symb{h}(\symb{b},\symb{c}) \to \symb{g}(\symb{a}), ~~
\symb{a} \to \symb{b}, ~~
\symb{a} \to \symb{c} 
~\}
\]
This TRS is innermost terminating, but not terminating because we have an infinite non-innermost rewrite sequence
$\symb{g}(\symb{a}) \to_{\cRghabc} \symb{h}(\symb{a},\symb{a}) \to_{\cRghabc} \symb{h}(\symb{b},\symb{a}) \to_{\cRghabc} \symb{h}(\symb{b},\symb{c}) \to_{\cRghabc} \symb{g}(\symb{a}) \to_{\cRghabc} \cdots$.
There is no innermost rewrite sequence that simulates the infinite rewrite sequence because, to apply the second rule, $\symb{a}$ of $\symb{g}(\symb{a})$, which is duplicated in applying the first rule to $\symb{g}(\symb{a})$, should be reduced after the application of the first rule to $\symb{g}(\symb{a})$, but $\symb{a}$ of $\symb{g}(\symb{a})$ must be reduced first under the innermost strategy.
The duplicated redex $\symb{a}$ must be reduced to $\symb{b}$ and $\symb{c}$, respectively, to apply the second rule:
$\symb{h}(\symb{a},\symb{a}) \to_{\cRghabc} \symb{h}(\symb{b},\symb{a}) \to_{\cRghabc} \symb{h}(\symb{b},\symb{c})$.
\end{example}

Our two assumptions ``right-linearity'' and ``being an overlay system'' are directly used in proving \Cref{lem:OL-RL-ISN-implies-SN}, which is the key claim for the results in this paper.
For the induction step $f(s_1,\ldots,s_n)\sigma \mathrel{\to_{>\varepsilon,\cR}^*} f(\ell_1,\ldots,\ell_n)\theta \mathrel{\to_{\varepsilon,\cR}} r\theta \mathrel{\to_\cR^*} t$ in the proof,
right-linearity ensures linearity of $r$, which is necessary to apply the induction hypothesis:
we obtain a normalized substitution $\theta'$ such that
$f(\ell_1,\ldots,\ell_n)\theta \mathrel{\to_{>\varepsilon,\cR}^*} f(\ell_1,\ldots,\ell_n)\theta' \mathrel{\ito_\cR} r\theta' \mathrel{\ito_\cR^!} t$.
Then, for $s_\sigma \mathrel{\to_\cR^*} \ell\theta'$, by the induction hypothesis, we obtain 
substitutions $\sigma_1,\ldots,\sigma_n$ such that 
$s_i\sigma \mathrel{\to_\cR^*} s_i\sigma_i \mathrel{\ito_\cR^*} \ell_i\theta'$
and $\Dom(\sigma_i) \subseteq \Var(s_i)$ for all $i \in \{1,\ldots,n\}$.
Linearity of $s$ is used to construct a substitution $\sigma' = \bigcup_{i=1}^n \sigma_i$:
linearity of $s$ implies that $\Var(s_i) \cap \Var(s_j) = \emptyset$ for $i\ne j$ and thus, 
$\Dom(\sigma_i) \cap \Dom(\sigma_j) = \emptyset$;
therefore, $\bigcup_{i=1}^n \sigma_i$ is a substitution.
Linearity of $s$ can be replaced by local confluence of $\cR$;
since $s\sigma$ is terminating, the reduction starting from $s\sigma$ is confluent;
suppose that a variable $x$ appears both in $s_i$ and $s_j$ ($i\ne j$);
then, it follows from confluence that $x\sigma_i = x\sigma_j$;
therefore, $\bigcup_{i=1}^n \sigma_i$ is a substitution.
In summary, assuming local confluence of $\cR$ instead of right-linearity, \Cref{lem:OL-RL-ISN-implies-SN} holds and implies \Cref{thm:SIN-SN-for-LCR-overlay-systems}.

\section{Related Work}
\label{sec:related-work}

Innermost reduction is one of the most fundamental strategies, which models eager evaluation of programming languages.
Therefore, the innermost strategy has been the subject of various studies for a long time, and in recent years it has attracted renewed attention in new types of rewriting computation models as, e.g., \emph{almost-sure innermost termination} of probabilistic term rewriting~\cite{KDG23,KFG24,KG25}.
In~\cite{dPZ05}, two kinds of generalized innermost rewriting have been investigated and equivalence of termination of such generalized innermost rewriting are equivalent to innermost termination.
The \emph{non-dup-generalized innermost rewriting} in~\cite[Definition~5]{dPZ05} is very closed to the discussions in this paper in the sense that right-linearity is required to avoid duplication of reducible subterms of redexes.
Regarding termination, confluence, and reachability of innermost rewriting, some decidable classes have been investigated \cite{GH07,KS08rta,KSNKS09,God10,USS10,IOS19}.

The closest work of this paper is~\cite[Theorem~3.23]{Gra95} (\Cref{thm:SIN-SN-for-LCR-overlay-systems} in this paper).
Since the result predates the proposal of dependency pairs, it does not use chains in its proof, unlike this paper.
However, similar to how minimal chains trace infinite rewrite sequences, rewrite sequences corresponding to chains are extracted from infinite rewrite sequences.
As shown in \Cref{ex:RLOS-LCROS-incomparable}, the classes of ``right-linear overlay systems'' and ``locally confluent overlay systems'' are incompatible.
Coincidence of termination and innermost termination for right-linear overlay TRSs has been shown in an unpublished note~\cite[Theorem~5]{Sakai03note}.
The proof in \cite[Theorem~5]{Sakai03note} is based on the approach in \cite{Gra96}, and dependency pairs are not used, e.g., there is no similar result to our main claim (\Cref{thm:inf-chain-iff-inf-innermost-chain}) in~\cite{Sakai03note}.

As described in \Cref{sec:intro}, a potential application area for the results in this paper is the dependency framework~\cite{GTSK04,GTSKF06}.
\Cref{thm:inf-DP-chain-iff-inf-innermost-DP-chain} replaces the side condition ``local confluence'' of the DP processor in order to switch termination to innermost termination~\cite[Theorem~32]{GTSK04} by ``right-linearity''.
Innermost rewriting may be worse than standard rewriting in the sense of \cite[Definition~6]{vO07}---in term of the number of rewrite steps to normal forms (see~\cite[Theorem~12]{dPZ05}).
On the other hand, in proving termination of a TRS, innermost rewriting would be useful if termination and innermost termination of the TRS are equivalent.

\section{Conclusion}
\label{sec:conclusion}

In this paper, we first showed that for a right-linear overlay TRS $\cR$ and a set $\cP \subseteq \DP(\cR)$, there is no infinite minimal $(\cP,\cR)$-chain if and only if there is no infinite innermost minimal $(\cP,\cR)$-chain.
Then, using the claim, we showed that 
termination and innermost termination coincide for the class of right-linear overlay TRSs.
Local confluence is not a syntactic property, while there are some syntactic sufficient conditions (e.g., non-overlappingness) for local confluence.
On the other hand, both right-linearity and being overlay systems are syntactic properties, and hence it is not so expensive to decide whether a given TRS is a right-linear overlay system.
In future work, we plan to empirically evaluate the results in this paper from the perspective of increasing the proof power for termination of TRSs using the switching processor mentioned in the last paragraph of \Cref{sec:intro}.
A more precise comparison with the non-dup-generalized innermost rewriting in~\cite[Definition~5]{dPZ05} is also an interesting further direction of this research.

\bibliography{mybiblio}

\appendix
\crefalias{section}{appsec}

\section{Proof of \Cref{lem:OL-RL-ISN-implies-SN}}
\label{sec:LemOlRlIsnImpliesSn}

We show a more general claim than \Cref{lem:OL-RL-ISN-implies-SN}, which considers a terminating term $s_0$ in order to use well-founded induction on ${\to_\cR}\cup{\rhd}$ over terms reachable from $s_0$.
\begin{restatable}{lemma}{AuxLemOlRlIsnImpliesSn}
\label{auxlem:OL-RL-ISN-implies-SN}
Let $\cR$ be a right-linear overlay TRS over a signature $\Sigma$, 
$s_0$ be a terminating term in $T(\Sigma,\cV)$,
$s$ be a linear term in $T(\Sigma,\cV)$,
$t$ be a normal form of $\cR$ over $\Sigma$ (i.e., $t \in \NF_\cR(\Sigma,\cV)$),
and
$\sigma$ be a substitution.
If $s_0 \mathrel{({\to_\cR}\cup{\rhd})^*} s\sigma \mathrel{\to_\cR^!} t$, 
then there exists a substitution $\sigma'$ such that
\begin{itemize}
    \item
    $\Dom(\sigma') = \Dom(\sigma|_{\Var(s)})$,
    \item
    $x\sigma \mathrel{\to_\cR^*} x\sigma' \in \NF_\cR(\Sigma,\cV)$ for all variables $x \in \Var(s)$
    (i.e., $s\sigma \mathrel{\to_\cR^*} s\sigma'$),
        and
    \item 
    $s\sigma' \mathrel{\ito_\cR^!} t$.
\end{itemize}
\end{restatable}
\begin{proof}
Since $s_0$ is terminating w.r.t.\ $\cR$,
the binary relation $\{ (s,s') \mid s_0 \mathrel{\to_\cR^*} s \mathrel{\to_\cR} s' \}$ is well-founded.
Let 
${\SuccST} = \{ (s',s'') \mid s_0 \mathrel{({\to_\cR}\cup{\rhd})^*} s'
\mathrel{({\to_\cR}\cup{\rhd})^+} s''
\}$.
Then, it follows from~\cite[Lemma~7.2.4]{Ohl02} that $\SuccST$ is well-founded.
We prove the claim by induction on $\SuccST$, in a manner similar to that of~\cite[Lemma~3.4]{SOS03}.

We first consider the case where $s$ is a variable $x$. 
Let $\sigma' = \{ x \mapsto t \}$.
Then, we have that $x\sigma = s\sigma \mathrel{\to_\cR^*} t = x\sigma' \in \NF_\cR(\Sigma,\cV)$
and $s\sigma' = x\sigma' = t$.
Therefore, the claim holds.

Next, we consider the remaining case where $s$ is not a variable.
Let $s = f(s_1,\ldots,s_n)$.
We make a case analysis depending on whether $s\sigma \mathrel{\to_\cR^!} t$ includes a topmost step $\to_{\varepsilon,\cR}$ or not.
\begin{itemize}
    \item Case where $s\sigma \mathrel{\to_{>\varepsilon,\cR}^!} t$.
    In this case, $t$ is of the form $f(t_1,\ldots,t_n)$ and we have that
    $s_i\sigma \mathrel{\to_{\cR}^!} t_i$ for all $1 \leq i \leq n$.
    Let $\sigma_i = \sigma|_{\Var(s_i)}$.
    Since $s$ is linear and $\Dom(\sigma|_{\Var(s)}) \subseteq \Var(s)$, 
    we have that $\Dom(\sigma_i) \cap \Dom(\sigma_j) = \emptyset$ for each $i,j \in \{1,\ldots,n\}$ with $i \ne j$.
    Note that $\sigma|_{\Var(s)} = \bigcup_{i=1}^n \sigma_i$.
    Since $s_i$ is a proper subterm of $s$, $s_i$ is linear and thus $s_i\sigma_i$ is a proper subterm of $s\sigma = f(s_1\sigma,\ldots,s_n\sigma)$.
    Thus, we have that $s_0 \mathrel{({\to_\cR}\cup{\rhd})^*} s\sigma \mathrel{\rhd} s_i\sigma_i$ (i.e., $s\sigma \SuccST s_i\sigma_i$).
    By the induction hypothesis, 
    for each $i \in \{1,\ldots,n\}$,
    there exists a substitution $\sigma'_i$ such that 
    \begin{itemize}
        \item $\Dom(\sigma'_i) = \Dom(\sigma_i|_{\Var(s_i)}) = \Dom(\sigma_i)$,
        \item $x\sigma \mathrel{\to_\cR^*} x\sigma'_i \in \NF_\cR(\Sigma,\cV)$ for all variables $x \in \Var(s_i)$,
            and
        \item $s_i\sigma'_i \mathrel{\ito_\cR^!} t_i$.
    \end{itemize}
    Since $\Dom(\sigma_i) \cap \Dom(\sigma_j) = \emptyset$ for each $i,j \in \{1,\ldots,n\}$ with $i \ne j$,
    we let $\sigma' = \bigcup_{i=1}^n \sigma'_i$, which is a substitution such that
    $\Dom(\sigma') = \Dom(\sigma|_{\Var(s)})$.
    Then, we have that 
    \begin{itemize}
        \item $x\sigma \mathrel{\to_\cR^*} x\sigma' \in \NF_\cR(\Sigma,\cV)$ for all variables $x \in \Var(s)$,
            and
        \item $s\sigma' = 
        f(s_1\sigma',\ldots,s_n\sigma')
        \mathrel{\ito_\cR^!}
        f(t_1,\ldots,t_n)
        = t$.
    \end{itemize}
    Therefore, the claim holds.    

    \item Case where $s\sigma \mathrel{\to_{>\varepsilon,\cR}^*} \cdot \mathrel{\to_{\varepsilon,\cR}} \cdot \mathrel{\to_\cR^!} t$.
    The overview of the proof for this case can be seen in \Cref{fig:overview-of-main-case-of-Auxiliary-Lemma}.
    Assume that 
    \[
    s\sigma = f(s_1\sigma,\ldots,s_n\sigma) \mathrel{\to_{>\varepsilon,\cR}^*} f(s'_1,\ldots,s'_n)
    =
    \ell\theta
    \mathrel{\to_{\varepsilon,\cR}}
    r\theta
    \mathrel{\to_\cR^!}
    t
    \]
    where
    $\ell \to r \in \cR$ and $\Dom(\theta) = \Var(\ell,r)$.
    Let $\theta_{\ell\setminus r}$ and $\theta_r$ be substitutions such that
    \begin{itemize}
        \item $\Dom(\theta_{\ell\setminus r}) = \Var(\ell) \setminus \Var(r)$,
        \item $\Dom(\theta_{r}) = \Var(r)$,
            and
        \item $\theta = \theta_{\ell\setminus r} \cup \theta_r$.
    \end{itemize}
    Since $\cR$ is right-linear, $r$ is linear.
    By the induction hypothesis, 
    there exists a substitution $\theta'_{r}$ such that 
    \begin{itemize}
        \item $\Dom(\theta'_r) = \Dom(\theta_r|_{\Var(r)}) = \Dom(\theta_r)$,
        \item $x\theta_r \mathrel{\to_\cR^*} x\theta'_r \in \NF_\cR(\Sigma,\cV)$ for all variables $x \in \Var(r)$,
            and
        \item $r\theta'_r \mathrel{\ito_\cR^!} t$.
    \end{itemize}
    Since $s_0$ is terminating w.r.t.\ $\cR$, for any variable $x \in \Var(\ell) \setminus \Var(r)$, $x\theta_{\ell\setminus r}$ is terminating and has a normal form.
    Let $\theta'_{\ell\setminus r}$ be a substitution such that
    \begin{itemize}
        \item $\Dom(\theta'_{\ell\setminus r}) = \Dom(\theta_{\ell\setminus r})$,
            and
        \item for any variable $x \in \Dom(\theta'_{\ell \setminus r})$, $x\theta'_{\ell \setminus r}$ is a normal form of $x\theta_{\ell \setminus r}$, i.e., $x\theta_{\ell \setminus r} \mathrel{\to_\cR^*} x\theta'_{\ell \setminus r}$.
    \end{itemize}
    Let $\theta' = \theta'_{\ell \setminus r} \cup \theta'_r$.
    Then, $\theta'$ is a substitution such that
        $\ell\theta = \ell(\theta_{\ell\setminus r}\cup \theta_r) \mathrel{\to_\cR^*} \ell(\theta'_{\ell\setminus r}\cup\theta'_r) = \ell\theta'$.

\begin{figure}[t]
\hfil
\xymatrix@R=28pt@C=6pt{
s\sigma \ar@{}[r]|(.22){\mbox{$=$}}
& f(s_1\sigma,\ldots,s_n\sigma) \ar[rrrr]^\ast_(.6){>\varepsilon,\cR} \ar@{..>}[d]_\ast^(.75){>\varepsilon,\cR} 
\ar@{}[drrrr]|{\mbox{\raisebox{5pt}{\small I.H.\ for $s_i\sigma \to_\cR^* \ell_i\theta'$}}}
&&&& f(\ell_1\theta,\ldots,\ell_n\theta) \ar@{}[r]|(.77){\mbox{$=$}} \ar@{..>}[d]_\ast^(.75){>\varepsilon,\cR}
& \ell\theta \ar[rrr]_(.6){\varepsilon,\cR} 
&&& r\theta \ar[rrrr]^{!}_(.8){\cR} \ar@{..>}[d]_\ast^(.75){\cR}
&&&& t
\\
s\sigma' \ar@{}[r]|(.22){\mbox{$=$}}
& f(s_1\sigma',\ldots,s_n\sigma') \ar@{..>}[rrrr]^{\mathsf{i}~~\ast}_(.6){\cR} 
&&&& f(\ell_1\theta',\ldots,\ell_n\theta') \ar@{}[r]|(.78){\mbox{$=$}}
& \ell\theta' \ar[rrr]^{\mathsf{i}}_(.65){\varepsilon,\cR}
&&& r\theta' \ar@{..>}@/_15pt/[urrrr]^(.4){\mathsf{i}}|(.6){\raisebox{10pt}{\scriptsize !}}_(.8){\cR} \ar@{}[urrrr]|{\mbox{\raisebox{5pt}{\small I.H.}}}
&&&&& %
\\
}
\caption{An overview of the proof for the case where $s$ is not a variable and $s\sigma \mathrel{\to_{>\varepsilon,\cR}^*} \ell\theta \mathrel{\to_{\varepsilon,\cR}} r\theta \mathrel{\to_\cR^*} t$ in \Cref{auxlem:OL-RL-ISN-implies-SN}, where solid arrows represent assumptions and dotted arrows represent consequences.}
\label{fig:overview-of-main-case-of-Auxiliary-Lemma}
\end{figure}

    Since $\ell$ is not a variable and $f(s'_1,\ldots,s'_n) = \ell\theta$, $\ell$ is rooted by $f$ and 
    $\ell\theta \mathrel{\to_{>\varepsilon,\cR}^*} \ell\theta'$.
    Let $\ell = f(\ell_1,\ldots,\ell_n)$.
    We now show that $\ell_i\theta' \in \NF_\cR(\Sigma,\cV)$ for any $i \in \{1,\ldots,n\}$.
    We proceed by contradiction.
    Assume that $\ell_i\theta' \notin \NF_\cR(\Sigma,\cV)$ for some $i \in \{1,\ldots,n\}$.
    Then, there exist a position $p \in \Pos(\ell_i)$, a rule $\ell' \to r' \in \cR$, and a substitution $\delta$ such that
    $\ell_i|_p \notin \cV$
    and
    $\ell_i|_p\theta' = \ell'\delta$.
    Assume w.l.o.g.\ that $\Var(\ell,r) \cap \Var(\ell',r') = \emptyset$.
    Then, $\ell_i|_p$ and $\ell'$ are unifiable.
    $\ell' \to r'$ overlaps with $f(\ell_1,\ldots,\ell_n) \to r$ at position $p > \varepsilon$.
    This contradicts the assumption that $\cR$ is an overlay system.
    Thus, $\ell_i\theta' \in \NF_\cR(\Sigma,\cV)$ for any $i \in \{1,\ldots,n\}$
    and $\ell\theta' = f(\ell_1\theta',\ldots,\ell_n\theta') \mathrel{\ito_{\varepsilon,\cR}} r\theta'$.

    Since $s\sigma = f(s_1\sigma,\ldots,s_n\sigma) \mathrel{\to_{>\varepsilon,\cR}^*} 
    f(\ell_1\theta,\ldots,\ell_n\theta) 
    \mathrel{\to_{>\varepsilon,\cR}^*}
    f(\ell_1\theta',\ldots,\ell_n\theta')
    $, we have that
    $s_i\sigma \mathrel{\to_{\cR}^!} \ell_i\theta'$ for all $1 \leq i \leq n$.
    Let $\sigma_i = \sigma|_{\Var(s_i)}$.
    Since $s$ is linear and $\Dom(\sigma|_{\Var(s)}) \subseteq \Var(s)$, 
    we have that $\Dom(\sigma_i) \cap \Dom(\sigma_j) = \emptyset$ for each $i,j \in \{1,\ldots,n\}$ with $i \ne j$.
    Note that $\sigma|_{\Var(s)} = \bigcup_{i=1}^n \sigma_i$.
    Since $s_i$ is a proper subterm of $s$, $s_i$ is linear and thus $s_i\sigma_i$ is a proper subterm of $s\sigma = f(s_1\sigma,\ldots,s_n\sigma)$.
    Thus, we have that $s_0 \mathrel{({\to_\cR}\cup{\rhd})^*} s\sigma \mathrel{\rhd} s_i\sigma_i$ and hence $s\sigma \SuccST s_i\sigma_i$.
    By the induction hypothesis, 
    for each $i \in \{1,\ldots,n\}$,
    there exists a substitution $\sigma'_i$ such that 
    \begin{itemize}
        \item $\Dom(\sigma'_i) = \Dom(\sigma_i|_{\Var(s_i)}) = \Dom(\sigma_i)$,
        \item $x\sigma \mathrel{\to_\cR^*} x\sigma'_i \in \NF_\cR(\Sigma,\cV)$ for all variables $x \in \Var(s_i)$,
            and
        \item $s_i\sigma'_i \mathrel{\ito_\cR^!} \ell_i\theta'$.
    \end{itemize}
    Since $\Dom(\sigma_i) \cap \Dom(\sigma_j) = \emptyset$ for each $i,j \in \{1,\ldots,n\}$ with $i \ne j$,
    we let $\sigma' = \bigcup_{i=1}^n \sigma'_i$, which is a substitution such that
    $\Dom(\sigma') = \Dom(\sigma|_{\Var(s)})$.
    Then, we have that 
    \begin{itemize}
        \item $x\sigma \mathrel{\to_\cR^*} x\sigma' \in \NF_\cR(\Sigma,\cV)$ for all variables $x \in \Var(s)$,
            and
        \item $s\sigma' = 
        f(s_1\sigma',\ldots,s_n\sigma')
        \mathrel{\ito_\cR^*}
        f(\ell_1\theta',\ldots,\ell_n\theta')
        = \ell\theta'$.
    \end{itemize}
    Since $\cR$ is an overlay system, all $\ell_1\theta',\ldots,\ell_n\theta'$ are normal forms of $\cR$;
    suppose that $\ell_i\theta'$ is not a normal form of $\cR$;
    then, $\ell_i$ has a non-variable subterm $\ell'_i$ such that $\ell'_i\theta'$ is a redex of $\cR$;
    this contradicts the assumption that $\cR$ is an overlay system.
    Thus, the term $f(\ell_1\theta',\ldots,\ell_n\theta')$ is an innermost redex of $\cR$ and hence
    $s\sigma' \mathrel{\ito_\cR^*} \ell\theta' \mathrel{\ito_{\varepsilon,\cR}} r\theta' \mathrel{\ito_\cR^!} t$.
    Therefore, the claim holds.
\end{itemize}
\vspace{-10pt}
\end{proof}

\LemOlRlIsnImpliesSn*
\begin{proof}
Let $s_0$ be $s\sigma$.
Then, by assumption, $s_0$ is terminating w.r.t.\ $\cR$ and we have that $s_0 \mathrel{({\to_\cR}\cup{\rhd})^*} s\sigma$.
Therefore, the claim immediately follows \Cref{auxlem:OL-RL-ISN-implies-SN}.
\end{proof}

\end{document}